\documentclass[10pt,a4paper]{article}
\usepackage[utf8]{inputenc}
\usepackage[T1]{fontenc}
\usepackage{amsmath}
\usepackage{amstext}
\usepackage{amsfonts}
\usepackage{amssymb}
\usepackage{amsthm}
\usepackage{graphicx}
\usepackage{authblk}
\usepackage{color}

\theoremstyle{plain}

\newtheorem{Thm}{Theorem}

\newcommand{\NN}{\mathbb{N}}

\newcommand{\red}[1]{\textcolor{black}{#1}}

\author[a,b]{Phillip L. Wilson}
\affil[a]{School of Mathematics \& Statistics, University of Canterbury, Private Bag 4800, Christchurch 8140, New Zealand. phillip.wilson@canterbury.ac.nz.}
\affil[b]{Te P\={u}naha Matatini, New Zealand.}
\title{Quantum Immortality and Non-Classical Logic}

\date{}

\begin{document}
	\maketitle

	\begin{abstract}
		The \emph{Everett Box} is a device in which an observer and a lethal quantum apparatus are isolated from the rest of the universe. On a regular basis, successive \emph{trials} occur, in each of which an automatic measurement of a quantum superposition inside the apparatus either causes instant death or does nothing to the observer. From the observer's perspective, the chances of surviving $m$ trials monotonically decreases with increasing $m$. As a result, if the observer is still alive for sufficiently large $m$ she rejects any interpretation of quantum mechanics which is not the many-worlds interpretation (MWI), since surviving $m$ trials becomes vanishingly unlikely in a single world, whereas a version of the her will necessarily survive in the branching MWI universe. That is, the MWI is testable, at least privately. Here we ask whether this conclusion still holds if rather than a classical understanding of limits built on classical logic we instead require our physics to satisfy a computability requirement by investigating the Everett Box in a model of a computational universe using a variety of constructive logic, Recursive Constructive Mathematics. We show that although the standard argument sketched above is no longer valid, we nevertheless can argue that the MWI remains privately testable in a computable universe. 
	\end{abstract}

\section{Introduction}\label{sec:intro}

The famously accurate predictions obtained by calculating solutions to the Schr\"{o}dinger equation do not depend upon the interpretation of quantum mechanics. For a physicist who wants to distinguish between these interpretations, being told to ``shut up and calculate'' misses the point: she is interested in ontology (what is real) not epistemology (what we can know). But since the predictions of all interpretations are the same, given by an interpretion-free solution of the Schr\"{o}dinger equation, how can we distinguish between them experimentally? In particular, is the \emph{many-worlds interpretation\footnote{So named by Bryce DeWitt \cite{dewitt70,dewitt72}.}} (MWI) of Hugh Everett \cite{ev,everett57} distinguishable from other interpretations? More fundamentally, is it testable?

These questions were addressed by Tegmark in \cite{tegmark97} (see also references therein for more history of the questions). In particular, Tegmark argued that the MWI is testable, and can be distinguished from other interpretations under extreme conditions, at least for one observer. Tegmark's argument (see also \cite{tegbook}) has two main points:
\begin{description}
	\item[T1] the MWI is testable and can be distinguished from other interpretations through experiment;
	\item[T2] the same experiment which can distinguish the MWI from other interpretations should also be taken as evidence in support of the MWI \emph{for the experimenter herself}.\footnote{In the language of \cite{lewis00}, the MWI is \emph{privately} testable, if not \emph{publicly} testable.} 
\end{description}

Popular descriptions of the MWI talk of the universe splitting into $n$ branches (or \emph{worlds}) whenever a quantum measurement has $n$ possible outcomes. However, as emphasised by \cite{tegmark97}, this was never postulated by Everett himself. Indeed, the core of the MWI is simply that the wavefunction never collapses and so the universe can be taken to be governed by a single, objectively real, universal wave function. However, the language of a splitting or branching universe, or observer, is a useful tool, and since it is also a common one (see a brief discussion of its history in \cite{lewis00}) we will sometimes use that language herein.

One of the better-known thought experiments which claims to establish T1 and T2 is a variant of the classic Schr\"{o}dinger's cat experiment in which the cat, or rather a human in place of the cat, is the observer \cite{tegbook,tegmark97,lewis00}. The life of the observer depends upon the outcome of an automatic measurement of a qubit called a \emph{trial}. Many trials occur, one after the other on a regular basis. We describe in detail in \S\ref{sec:qi} how this \emph{Everett Box}\footnote{Named in homage to the genesis of these ideas in Everett's seminal thesis \cite{ev} even though Everett himself never formulated them.} is able to establish T1 and T2. The argument is a probabilistic one: though the observer might get lucky and survive a few trials, continued survival in a non-MWI universe is extremely unlikely. On the other hand, a version of the observer is guaranteed to survive with probability 1 in one branch of an MWI universe, guaranteeing the so-called \emph{Quantum Immortality} of that observer. From the private perspective of the observer, her continuing survival therefore counts as evidence in support of the MWI, establishing both T2 and T1, non-MWI options being rejected as being too unlikely.

We call the line of argument summarised above and presented in \S\ref{sec:qi} the \emph{Quantum Immortality Argument (QIA)}, and refer to its conclusion as \emph{Quantum Immortality}. Although Everett never formally defined this experiment, variants of it have been given independently by several authors \cite{squires86, tegmark97}.  Neither the argument nor its conclusion is  universally accepted. The QIA has been attacked from various directions, not least in terms of its real-world applicability, for instance around the definition of death (or at least of a discrete binary distinction between ``alive'' and ``dead'') --- see \cite{tegbook,tegmark97}. From the philosophical perspective we see critiques\footnote{Not all of the authors of these critiques reject the MWI even while  rejecting the QIA; see for example \cite{deutsch99}.} based on the classical philosophical problem of individual identity and its persistence, what it means to ``expect'' a subjective outcome like one's own death as opposed to predicting an objective event, the distinction between actual and probable events, and the meaning of probabilistic thinking in the MWI context\footnote{\red{Of particular note here is the lack of an ensemble or ``God's eye'' view of the experiment: there is no vantage point from which an external obsever can quantify overall outcomes.}}; see \cite{deutsch99,lewis00,papineau04,aranyosi12,sebens15,sep-qm-manyworlds}, and references therein. Everett himself anticipated some of these objections in \cite{ev}.

There are two other ways in which the QIA is critiqued, both of which are much more general in their scope. They concern (1) the role of infinity and the infinitesimal in physics, and (2) the role of the computable. A motivation for the first of these is that if we live in a finite universe which has existed for finite time, and if the fields, matter, time, and space of the universe are all discrete at sufficiently small scales, then we should reject all objects and arguments which employ the infinite and the infinitesimal. Such strictly \emph{finitist} theories include digital physics \cite{wheeler90}, cellular automata \cite{nks}, loop quantum gravity \cite{loop1}, and more besides --- see \cite{schmidhuber97} and references therein. The questionable role of infinity in the QIA has been highlighted by \cite{tegbook} amongst others.

The second critique, namely that our current theories of physics are non-computable, is the focus of the present work. Requiring a computable theory of physics is essentially the same as requiring all knowledge to be obtained through an algorithmic process in finite time. It is not the same as requiring only finite objects, but it does necessitate working within so-called non-classical logics, as we outline in detail below. The desirable quality of computability in the foundations of physics is not obtained by classical logic.

Thus while probabilities and probabilistic thinking have been highlighted as potential concerns with the QIA \cite{lewis00,papineau04}, and while the role of the infinite and the infinitesimal in physics have also been called into question in this context \cite{tegbook}, to our knowledge no-one has examined the argument from a computable perspective before, and in particular from within non-classical logic.

Here we show that the testability of the MWI and the subjective evidence in support of the MWI given by the QIA are based on a classical understanding of limiting behaviours of functions which need not hold in other, non-classical logics. In particular, we show that the Quantum Immortality Argument fails in a constructive logic called Recursive Constructive Mathematics, commonly referred to as RUSS, in which all results are computable. Within RUSS, we show that the existence of so-called \emph{pathological} probability distributions mean that we must reject the QIA. However, we are able to show through a new argument that a constructive version of the QIA holds even in a universe (or universes) governed by such non-classical logics, and thus that the MWI remains (privately) testable and distinguishable from other interpretations of quantum mechanics.

In \S\ref{sec:qi} we give a brief overview of the QIA, and in particular how the Everett Box implies T1 and T2. Next, in \S\ref{sec:monk} we define computability and outline the arguments in favour of requiring computability in theories of physics, before giving a summary of the main result from \cite{mjw} on which we base the principal argument in this paper. With this background we prove in \S\ref{sec:main} the Pathological Mortality Theorem, which shows that the QIA does not work in RUSS. However, in \S\ref{sec:restore} we present a constructive, computable proof that Quantum Immortality nevertheless holds in a universe whose logic is that of RUSS. We call this argument the \emph{computable Quantum Immortality Argument}. It shows that the MWI is testable in a universe governed by computable logic. Finally, we summarise and discuss our results in \S\ref{sec:conc}.

\section{Quantum Immortality}\label{sec:qi}

A conscious observer is placed in a box with a lethal quantum apparatus. The contents of the box are completely isolated from the rest of the universe. Although this thought experiment does not depend on the details of the lethal apparatus, a particularly clear example is given by \cite{tegmark97} and called the ``quantum gun''. The quantum gun consists of a gun coupled to a quantum system of a particle in a superposition of two states. At regular time intervals, a measurement of this qubit is made automatically, and if it is found to be in one state the gun fires a bullet, while if it is in the other it does not fire. After either firing or not firing, the quantum gun resets: a new superposition is set up and the memoryless process repeats\footnote{This slight variant of Tegmark's quantum gun of \cite{tegmark97} was given in \cite{tegbook}. Note that both of Tegmark's versions include an assistant in the box with the observer, but this is not a necessary feature of the setup and has been omitted through automation here.}. Each independent occurrence of this process we call a \emph{trial}.  We take this or a similar lethal setup to be indefinitely repeatable and to occur every second\footnote{The time interval is not important to the subsequent argument, other than to allow for many repetitions within a human lifetime.}. The apparatus and the observer are isolated from the rest of the universe, and this setup constitutes the Everett Box.

What is the experience of the observer? It is rather starkly illustrated by Tegmark's gun if we contrast the Everett Box with a similar experiment in which instead of being aimed at the observer the gun merely fires or does not fire depending on the measurement, and is aimed at a target in the box while the observer observes the experiment from within the box. In this case, the observer can expect to hear a random string of bangs and clicks: the bangs correspond to the gun firing, the clicks to it not firing and the equipment resetting. Over time the relative proportion of bangs and clicks will tend towards the relative likelihoods of those two outcomes. In the standard formulation, both outcomes occur with equal probability and thus the observer expects over time that 50\% of the sounds will be bangs, and 50\% clicks. The QIA is actually independent of these likelihoods, which need not be either equal or constant \cite{ev,everett57,tegbook}. It is such a general case that we consider in this paper.

The immediately preceding description is not that of the Everett Box, because the life and hence consciousness\footnote{Consciousness surviving death is not a part of the thought experiment.} of the observer does not depend upon the outcome of the measurement of the qubit. In the Everett Box, the gun is aimed at the observer in such a way that should it fire then death is certain and swift\footnote{Both conditions are necessary as outlined in \cite{tegmark97,lewis00}.}. In this case, what should the observer expect\footnote{\red{Assuming that she has first verified that her appartus is working by observing its operation without being in the line of fire.}}?

The answer depends upon which interpretation of quantum mechanics holds in our universe. If there is only one world, then for the majority of interpretations each trial involves the collapse of the wavefunction and a single outcome occurs for the observer: either she hears a ``click'' or she is instantly killed (and so hears nothing)\footnote{Note that even non-collapse interpretations such as Bohmian mechanics predict a single outcome for the observer.}. She might get lucky once, she might get lucky twice, but as time goes on and the number of trials increases, the odds of her surviving decrease exponentially.

If, however, there are many worlds, the totality of which contain all possible outcomes and histories, then by necessity there is always an observer alive after any number of trials. For example, after one trial there are two versions of the observer, the universe\footnote{Or at least, the observer.} having branched into two at the moment of the quantum measurement. In one world the observer heard ``click'' while in the other she died. After two seconds there are four worlds. In one of them, the observer's history shows that she heard ``click-click''. In another world, she heard ``click'' and then died on trial 2. In the third and fourth she died on trial 1\footnote{Here we assume that a predetermined number of trials occur regardless of whether the observer is alive. This assumption is not actually necessary for the arguments which follow, and can easily (though with some loss of elegance) be removed.}. After three trials there is an observer whose history is ``click-click-click'', and after any number, $m$, of trials there will always be one world in which the observer has survived to hear $m$ clicks. After a large number of trials there are many worlds, in all but one of which the observer is dead, but crucially there remains one living observer. Thus the subjective probability of surviving $m$ trials is 1 for any $m$, because there is a world in which the observer is still alive after any number of trials \cite{lewis00}. Note that the $2^m$ worlds ``created'' in this branching process are not used for quantifying the \emph{subjective} probability of 1 that the the observer survives --- at no point do we calculate subjective probabilities over a set of universes in only some of which does a subject exist \cite{lewis00}.

Thus from the observer's perspective\footnote{And from hers alone; an external observer opening the box after a pre-ordained number of seconds (trials) will almost certainly find a corpse within the box. This footnote is the only point in the present paper at which we quantify over a multiverse of branched universes.} this experiment an establish T1 and T2, though the stakes are high. The argument runs as follows. The chances of remaining alive after a large number of trials in a non-MWI universe is monotonically and exponentially decreasing because each trial is independent of the preceding trials. Thus at some point the probability of being alive will be lower than some threshold at which the still-alive observer can reject any non-MWI interpretation purely on the grounds of the low probability of such a sequence of events occurring. This is a standard experimental approach, and as usual the threshold $\epsilon\ll 1$ could be set to the traditional $5\sigma$-level, or indeed to a level of any stringency based on any criterion\footnote{It is also possible to formulate a threshold based on a Bayesian analysis.} due to the monotonically decreasing dependence on $m$ of the probability of remaining alive. Furthermore, with each subsequent survived trial the confidence in rejecting any non-MWI interpretation increases. Of course, if we do not live in an MWI universe then the experiment simply kills the observer within a short time. She does not know that she does not live in an MWI universe, but neither does she know anything ever again.

In more rigorous terms, in non-MWI interpretations the probability of being dead after $m$ trials, $P(m)$, in the standard presentation in which the probability of death at each trial is $50\%$, is simply
\[P(m) = 1- \left(\frac{1}{2}\right)^m\]
which tends to unity as $m\to \infty$. The same conclusion holds regardless of the probability of staying alive on trial $k$, which we denote $p_k$. In this case,
because $p_k< 1$  for all $k$ we still have
\begin{equation}\label{eq:class}
P(m)=1-\prod_{k=1}^m p_k \to 1 \quad \text{as} \quad m\to\infty.
\end{equation}
While (\ref{eq:class}) will remain true throughout this paper, we will see that in the computable logic RUSS we can no longer use it to conclude that the observer necessarily must expect to be dead after any finite number of trials. First, we must review what it means to be computable.

\section{Computability and The Infinite Monkey Theorem}\label{sec:monk}

\subsection{Computability and Logic}\label{subsec:comp}

A problem is said to be \emph{computable} if it can be solved in an effective manner, which can be more formally defined in a number of models of computation \cite{cooper04,copelandbook,bridgescomp}. Loosely speaking, computable problems are those which can be solved algorithmically in finite time. The major milestone in computability theory is the Turing-Church thesis identifying computable functions on the natural numbers with functions computable on a Turing machine \cite{bridgescomp,sep-mathematics-constructive, sep-recursive-functions,copelandbook}.

It is not simply the rise in computer simulations, nor the ``shut-up-and-calculate'' instrumentalist approach to physics \cite{mermin04}, which have led some authors to suggest that computability should be a requirement for our theories of physics \cite{zuse69,loop1,schmidhuber97,thooft99,fredkin03,lloyd2005theory,nks}. It is instead the notion of the \emph{effective method} embedded in computability that is important. A method is called effective for a class of problems when it comprises a finite set of instructions which can be followed by a mechanical device\footnote{The idea here is not that they must be followed by such a device, but that even a human following them needs no \emph{ingenuity} in order to derive a correct answer.}, that these instructions produce a correct answer, and that they finish after a finite number of steps \cite{copelandbook}.

\red{From a philosophical perspective, t}he desirability of computability in physics is therefore a product of a desire to know, and a belief that the universe is ultimately comprehensible to us. The reasoning in the syllogism goes that if we accept the two premises that (1) the universe is entirely comprehensible to the human mind, and (2) there is nothing extra-computational happening in the human mind, then we must accept the conclusion that physics is necessarily computable.

\red{Moreover, from a practical point of view, the quantum calculations often used to verify observations are based on the classical solution to systems of partial differential equations. Both the systems themselves, with their corresponding boundary and initial conditions, and the numerical soltuions, are formulated and numerically approximated using methods relying ultimately on classical logic --- see for example the discussion of the differential operator in \cite{bridges97}, and other issues outlined in \cite{sep-mathematics-constructive}.}

However, our current theories of physics are not computable, built as they are on classical mathematical ideas which in turn rely on classical, non-computable, logic \cite{sep-mathematics-constructive}. There are two issues here. The first concerns the notion of infinity and the related notion of continuity. Infinities abound in our physical theories, whether they are in limiting behaviours (as examined in non-classical logics in the present paper) or in the related idea of continuous matter or continuous fields. In the latter case, even though we seem to know that neither matter nor fields are continuous in our universe, we treat the ``gap'' between our continuous theories and discrete nature as being essentially a rounding error: the high accuracy of predictions made with the (presumptively Platonic) continuous theories is because our universe is approximately continuous. It is, after all, perhaps only discrete below the Planck length, or on time scales shorter than the Planck time.

The second issue, and the one that concerns us in this paper, is the notion of the underlying logic of the universe. Classical logic is not computable, relying as it does on non-computable notions such as the Law of Excluded Middle (LEM) and omniscience principles \cite{varieties,sep-mathematics-constructive}. Why should we work with a logic that does not allow for computability if we wish our physics to be computable? One answer is similar to the response to continuity and infinity: because this logic works, to an astonishing degree \cite{wilson18}. A second response is simply to reject the second premise given above. Perhaps there is something extra-computational happening within the human mind\footnote{This perspective overlaps somewhat with the notion of \emph{hypercomputation} \cite{copeland02,copeland04}}. This is consistent with a robustly Platonic vision of the universe. If mathematical objects exist in a Platonic realm of forms to which our minds (somewhat mysteriously) have access, then the necessity of computability can be rejected. This is also consistent with the view above that our physical universe is only an (albeit excellent) approximation to one of Platonic forms.  

If however we insist with the authors above that our logic must be computable, then we necessarily have to work with non-classical logics which are computable. In particular, we should work within so-called \emph{constructive} interpretations of logic \cite{sep-mathematics-constructive}, in which the classical interpretations of disjunction and existence are rejected in favour of constructive ones. For example, the quantifier ``there exists'' becomes ``we can construct (that is, give an effective method for defining) an object for which the given statement is true''. There are several varieties of constructive mathematics \cite{varieties}. It should be noted that not all varieties reject notions of infinity. Bishop's Constructive Mathematics (referred to as BISH) \cite{bish}, for example, admits many classical mathematical objects which rely on infinities and continuity, but insists that proofs using these objects must proceed constructively (and are therefore computable). This illustrates the important distinction between the \emph{epistemological constructivism} of BISH which remains agnostic on the ontology of mathematical objects, and the \emph{ontological constructivism} of other varieties of constructive mathematics which insist that both objects and proofs (procedures) must be computable \cite{sep-mathematics-constructive,varieties}. It has been said that computable mathematics is simply mathematics done with intuitionistic logic \cite{sep-mathematics-constructive}.

In order to subject the QIA to a strong scrutiny in a non-classical logic, we here choose an ontologically constructive variety of constructive logic, Recursive Constructive Mathematics, RUSS \cite{sep-mathematics-constructive}.  RUSS is a constructive version of recursive function theory, in which functions on the natural numbers are defined recursively. Essentially, RUSS takes the classical recursive analysis in the tradition of Turing and Church but uses only intuitionistic logic.

We remind the reader that this paper is not an argument in support of constructive methods in physics in general, neither of RUSS in particular. (The interested reader is referred to \cite{hellman93,nielsen97,bridges97,bridges99,bridgessvozil00} amongst others.) Rather, it seeks to examine the QIA in a computable context, in anticipation either of reaching a testable implication that the universe is computable, or of recasting the QIA in a computable form. In the following subsection, we briefly outline the theorems of a recent work in computable probability based on RUSS which will be central to the argument of this paper.

\subsection{The Infinite Monkey Theorem}\label{subsec:imt}

Working in RUSS, \cite{mjw} proved a seemingly counter-intuitive theorem, which we call here the \emph{Infinite Monkey Theorem (IMT)}. To state the IMT we first need some notation. The IMT was written in the playful language of the famous aphorism that a large enough group of monkeys with typewriters will reproduce the complete works of Shakespeare, but as is made clear in \cite{mjw}, the IMT is really about computable probability distributions, as indeed is our focus in the present paper.

Retaining the metaphor of \cite{mjw}, we work in an alphabet $A$ (of size $|A|$, including punctuation) and call a \emph{$w$-string} any string of characters of length $w \in \NN$. For example, ``\emph{Everett}'' is a 7-string over the alphabet $\{E,e,r,t,v\}$. Each monkey works on a computer keyboard with $|A|$ unique keys and each monkey types a $w$-string in finite time. We define $M$ to be an infinite, enumerable set of monkeys (the \emph{monkeyverse}), and for any $m \in \NN$ the \emph{$m$-troop} of monkeys to be the first $m$ monkeys in $M$.  We then have
\begin{Thm}[Infinite Monkey Theorem]\label{thm:imt}
	Given a finite target $w$-string $T_w$ and a positive real number $\epsilon$, there exists a computable probability distribution on $M$ of producing $w$-strings such that:
	\begin{enumerate}
		\item[\textrm{(i)}] the classical probability that no monkey in $M$ produces $T_w$ is $0$; and
		\item[\textrm{(ii)}] the probability of a monkey in \emph{any} $m$-troop producing $T_w$ is less than $\epsilon$.
	\end{enumerate}
\end{Thm}
\cite{mjw} established an even stronger, target-free version of this theorem, which requires only a knowledge of $w$, not of $T_w$.

The theorem and its proof are computable. The theorem shows that while it is classically true that it is impossible that no monkey reproduces the works of Shakespeare (part (i)), it is possible to construct a so-called \emph{pathological} probability distribution on the monkeyverse such that the chances of actually finding the monkey that does so can be made arbitrarily small (part (ii)). The key point in part (ii) is that this is true for \emph{any} finite $m$-troop of monkeys; the pathological distribution does not require knowledge of the size of the $m$-troop, it is simply pathological for all finite sets.

The monkeys correspond to any finite black-box process occurring in finite time. The general conclusion drawn in \cite{mjw} is that in a computable universe the space of all possible probability distributions on enumerable sets contains a non-empty set of pathological distributions for which the IMT holds. This is in contradistinction to a universe governed by classical logic in which the IMT does not hold. It is this distinction that we exploit in the remainder of the paper, by examining the impact of the existence of pathological distributions on the enumerable set of trials in the Everett Box.

\section{Pathological Distributions Imply the Rejection of the Quantum Immortality Argument}\label{sec:main}

With the notation from \S\ref{sec:qi} we can state that the probability $P(m)$ of dying within $m$ trials is given by
\begin{equation}
P(m) = 1-\prod_{k=1}^m p_k
\end{equation}
for any $m\in\NN$, where $p_k$ is the probability of not dying on trial $k$. Note that here, in contrast to the manner in which the Everett Box is normally described, but in keeping with the more general case which Everett himself allowed for in \cite{ev}, we consider a quantum apparatus with variable probabilities at each trial. The subsequent argument does not depend upon the unknowability in advance of the probability of death at each trial, $1-p_k$; after all, in the standard formulation, $p_k=0.5$ for all $k$. Whatever the distribution of values of $p_k$, the observer cannot predict in advance whether she lives or dies on trial $k$, and her fate is determined purely by the unknowable quantum state of the apparatus. We can now state the following theorem.
\begin{Thm}[Pathological Mortality Theorem]\label{thm:pqi}
	The QIA fails in RUSS.
\end{Thm}

\begin{proof}
	The QIA relies upon the observer rejecting non-MWI interpretations once  $1-P(m)$, the probability of her surviving $m$ trials, drops below some threshold. She can always bound the number of trials required to drop below this threshold even if she is ignorant of the distribution of $p_k$.
	
	However, while classically $1-P(m) \to 0$ as the number of trials tends to infinity, there is a computable probability distribution on the trials such that the probability that the observer is alive after \emph{any} finite number of trials is arbitrarily close to 1. This follows directly from the proof of the IMT in \cite{mjw}. In particular, we place the objects of the IMT and the objects of the Pathological Mortality Theorem (PMT)  in one-to-one correspondence as outlined in the following table.
\begin{center}
	\begin{tabular}{ | c || p{5cm} | p{5cm} | }
		\hline
		 & IMT & PMT \\ \hline \hline
		$p_k$ & probability that $k^\text{th}$ monkey fails to reproduce Shakespeare & probability of not dying on $k^\text{th}$ trial \\ \hline
		$P(m)$ & probability that $m$-troop does reproduce Shakespeare & probability of dying within $m$ trials  \\
		\hline
	\end{tabular}
\end{center}

Thus while classically it remains true that the observer's probability of being alive after $m$ trials tends to 0 as $m$ tends to infinity, the classical interpretation of that result as being that after a certain \emph{finite} number of trials the probability of the observer being alive should be so small that she should be surprised at remaining alive and reject any non-MWI interpretation is not true in a computational sense, in which that probability can remain arbitrarily close to 1 for \emph{any} finite number of trials.
\end{proof}

There is an apparent contradiction between the classical probability of remaining alive tending to zero while it remains arbitrarily close to unity for \emph{any} finite number of rials. However, as outlined in \cite{mjw}, it is important to note that the apparent contradiction here is only between the classical notion of the limit and the existence of computable pathological distributions within RUSS; we are deliberately comparing results from non-commensurate logical systems in order to show that classical logic may lead us astray in a computable universe.

The PMT says that in a universe run on computable logic\footnote{We take this to be equivalent to the requirement for computability in the logic we use in our theories of physics.} the QIA is no longer valid, and therefore if our universe is one governed by computable logic the sole argument which claims to show that MWI is testable is no longer valid. If the quantum apparatus happens to be governed by a pathological distribution then it is no longer unlikely that the observer remains alive after any finite number of trials, since that likelihood can remain arbitrarily close to 1.  As a result, since the observer can never know for sure that she is not in a pathological distribution, she cannot surely state that remaining alive after any finite number of trials is unlikely and so she can never reject non-MWI interpretations of quantum mechanics. However, in the next section we formulate a new version of the QIA which hold even in a computable universe and even with the existence of the PMT.

\section{Quantum Immortality Restored}\label{sec:restore}

We state our main result as a theorem.

\begin{Thm}[computable Quanum Immortality Argument, cQIA]\label{thm:restored}
	The Everett Box implies that the MWI is testable and provides private evidence in support of the MWI even in a computable universe modelled by RUSS.
\end{Thm}

\begin{proof}
 Suppose an observer in an Everett Box in a RUSS-universe has survived many trials. There are two situations to consider depending on whether the probability distribution on the trials is pathological or not. To reiterate, the observer does not know and has no way of knowing which situation holds.
 
 First, if the probability distribution is not pathological then the standard, classical-logic QIA holds, and the observer concludes that she must reject all non-MWI interpretations of quantum mechanics.
 
 On the other hand, suppose that the distribution is pathological. Since Theorem 5 of \cite{mjw} showed that such distributions are vanishingly rare then the observer must reject all non-MWI interpretations since in a single world the odds of being in such a distribution are vanishingly small, whereas in the MWI there will always be a branch of the observer alive in a pathological distribution.
\end{proof}

To state the proof in other terms, we note that in a classical universe, so the original QIA goes, the observer rejects non-MWI interpretations because the odds of surviving repeated trials are so low, whereas in a RUSS computable universe she rejects non-MWI interpretations for the same reason if she happens to be in a non-pathological distribution, or because the odds of being in a pathological situation where the PMT holds in a single world are also vanishingly low. The observer does not need, therefore, any knowledge of whether she is in such a pathological experiment since either way she must reject non-MWI interpretations on the same basis: namely, the unlikelihood of being in that situation if there is only one world.

\section{Discussion}\label{sec:conc}

Before a broader discussion, we once again reiterate that at no point does the observer need to quantify probabilities over other worlds. She merely compares her probability of surviving $m$ trials \emph{given there is one world} with a threshold set by some standard criterion such as the common $5\sigma$-significance. In the cQIA, she makes two comparisons. In the first, she compares to the threshold the likelihood of there being only one world and in that world her equipment is governed by a non-pathological distribution. In the second, she compares to the threshold the likelihood of there being only one world and that in her world her equipment is governed by a pathological distribution. In both cases, she can easily pass any threshold by running a sufficient number of trials, and so she rejects the null hypotheses, thereby rejecting the existence of a single world. This she takes as private evidence for the MWI, which therefore remains testable even in a computable universe.

Therefore, we have argued that the Everett box is testable and provides evidence in support of the MWI interpretation of quantum mechanics even when the computability requirement is added to physics through employing RUSS, a constructive, computable logic. Our main point is therefore that those whose rejection of the MWI depends on some future recasting of physics in a computable form do not have that option if RUSS is the correct logic on which to base physics. We have shown, in fact, that a computable version of the QIA which we call the cQIA still holds in at least one computable logic. The case to be made against the MWI therefore must be stronger than has previously been appreciated.

This naturally raises the question as to whether a similar argument holds in other computable logics\footnote{In which case, perhaps the ``c'' in cQIA would deserve capitalization.}. For example, can we reproduce the argument in BISH, a computable logic which preserves most classical mathematical objects, including some of those which involve either infinity or continuity? What about in other logics which do not allow for such objects? And of course, what happens in a completely finitist universe? 

We have two final points to make. The first is to point out that although the argument here is given in a quantum context, the argument behind the proof of the PMT works in any situation, quantum or otherwise, in which the probability of an event occurring tending to 1 in the limit of an infinite sequence of trials is taken to mean that the probability of that event not having happened in any finite sequence of trials necessarily tends to 0. In RUSS, this is not true.

Finally, we remind the reader that the MWI does not assume that new universes are ``created'' at each quantum decision --- it takes the language of ``branching'' as a useful metaphorical tool rather than a literal description of reality, a reality simply described by a wavefunction which never collapses (see \cite{tegmark97} and references therein). In fact, Tegmark has argued (\cite{tegmark97,tegbook} and elsewhere) that a non-collapsing wavefunction plus decoherence strongly implies a platonic ontology in which the wavefunction and other mathematical objects are truly real, in fact constitute the only real things, and that the categories of human perceptions of the world are not to be taken literally, since they are prejudiced by our evolutionary history (see also \cite{hoffman}). This ``Mathematical Universe Hypothesis'' is an extreme form of platonism. We note in closing that despite popular arguments to the contrary, a platonic ontology is entirely compatible with a constructive epistemology --- see \cite{sep-mathematics-constructive}.


\red{The author wishes to thank the referees of an earlier version of this paper for their insightful comments, the responses to which have certainly improved the paper.}

\bibliographystyle{alpha}
\bibliography{qi}

\end{document}